\documentclass[conference]{IEEEtran}
\usepackage{cite}

\ifCLASSINFOpdf
 \usepackage[pdftex]{graphicx}
 \graphicspath{{../pdf/}{../jpeg/}}
\else
 \usepackage[dvips]{graphicx}
 \graphicspath{{../eps/}}
\fi
\usepackage[cmex10]{amsmath}
\usepackage{amsfonts}
\usepackage{algorithmic}

\usepackage{array}
\usepackage{amsthm}
\usepackage{subfigure}
\usepackage{mathtools}
\newtheorem{theorem}{Theorem}
\newtheorem{lemma}[theorem]{Lemma}

\DeclareMathOperator*{\argmin}{arg\,min}

\begin{document}
%
\title{A Novel Antenna Selection Scheme for Spatially Correlated Massive MIMO Uplinks with Imperfect Channel Estimation}

\author{\IEEEauthorblockN{De Mi,
Mehrdad Dianati,
Sami Muhaidat}
\IEEEauthorblockA{Centre for Communication Systems Research (CCSR)\\
University of Surrey, Guildford, UK\\
Email: {\{d.mi, m.dianati, s.muhaidat\}@surrey.ac.uk}
\and
\IEEEauthorblockN{Yan Chen}
\IEEEauthorblockA{Huawei Technologies, Co. Ltd.\\
Shanghai, China\\
Email: eeyanchen@huawei.com}}}


%



\maketitle

\begin{abstract}
We propose a new antenna selection scheme for a massive MIMO system with a single user terminal  and a base station with a large number of antennas. We consider a practical scenario where there is a realistic correlation among the antennas and imperfect channel estimation at the receiver side. The proposed scheme exploits the sparsity of the channel matrix for the effective selection of a limited number of antennas. To this end, we compute a sparse channel matrix by minimising the mean squared error. This optimisation problem is then solved by the well-known orthogonal matching pursuit algorithm.  
Widely used models for spatial correlation among the antennas and channel estimation errors are considered in this work. Simulation results demonstrate that when the impacts of spatial correlation and imperfect channel estimation introduced, the proposed scheme in the paper can significantly reduce complexity of the receiver, without degrading the system performance compared to the maximum ratio combining.
\end{abstract}


%
\IEEEpeerreviewmaketitle

\section{Introduction}

Multiple-input multiple-output (MIMO) system that employs a large number of antennas, known as massive MIMO systems \cite{T1, MSP2013}, has been recently proposed as a potential technique, for next generation wireless communication systems \cite{MMNG, Five5G}. Specifically, in \cite{T1}, more than 10 fold throughput improvement by the massive MIMO has been suggested to be achievable  compared to LTE (Long Term Evolution). Despite such potential, in practice, deployment of massive MIMO systems is hindered by practical challenges. Firstly, there is an inherent problem of spatial correlation among the antennas due to lack of possibility of allowing sufficient spacing among the antennas \cite{MSP2013, LargeMISO, MMIMOULDL}. Specifically, for the uplink transmission, the antenna correlation could be significant among all the antennas at the base station (BS) side due to the space limitation. In \cite{MSP2013}, it is shown that correlation among antennas can result in nearly negligible achievable capacity gains. In addition, the channel state information (CSI) is required in order to perform post processing at the receiver side for most of receiver implementations. Prior investigations show that imperfection in channel estimation can significantly degrade the system performance, especially for massive MIMO systems \cite{LargeMISO, MMIMOULDL}. 

For the massive MIMO uplinks, the impact of spatial correlation and imperfect channel estimation have been investigated in \cite{T1, MSP2013, LargeMISO, MMIMOULDL, MMIMOUL, rzg, rws}. More specifically, in \cite{T1}, it is shown that uplink combining schemes, such as maximum ratio combining (MRC), can have a reasonable performance, with knowledge of CSI for all antenna branches. However, the price to pay for such gain is the significantly increased implementation overhead and the complexity of the transceiver design for massive MIMO systems \cite{EEAS, AS2013}. In \cite{EEAS}, it is argued that cost-efficient antenna selection strategies can be employed to reduce the complexity and overhead of implementation, as well as to effectively maintain a reasonably high performance. Selection combining (SC) for uplink has been extensively studied in the literature such as \cite{r18, r21}, in the context of conventional MIMO systems. For example, the effect of imperfect channel estimation on the SC systems is investigated in \cite{r21}, but not for the large scale antenna systems.

An analysis of the MRC in massive MIMO uplinks under imperfect channel estimation is given in \cite{T1}. Exploiting sparsity, the work in \cite{r3} investigates antenna/relay selection for MIMO channels. However, this work of \cite{r3} does not take into account spatial correlation among antennas, as well as the impacts of imperfect CSI acquisition. Considering the spatial correlation and imperfect channel estimation, spatially correlated channel models in \cite{r16, r17, LargeMISO, r20} are considered as a good approximation for large scale antenna correlation, and channel estimation errors in \cite{r11, LargeMISO, r20} are applied to effectively model the imperfection caused by the practical channel estimation schemes.

The main contribution of this work is to propose an effective antenna selection combining scheme for spatially correlated single-user massive MIMO uplinks under the imperfect channel estimation, by applying a sparsely structured channel matrix at the BS side. The basic idea is to reduce the effective number of antennas that are used for combining. Consequently, the resulting effective channel matrix becomes sparse, in the sense that the corresponding entries to non-selected antennas are set to zero. This sparse channel matrix can be obtained by some approximation techniques that will be discussed later in this paper. Simulation results indicate that the proposed scheme can significantly reduce implementation complexity and overhead, e.g., it is shown that that only less than half number of antennas are required to achieve a performance level that is comparable to MRC scheme, when the effects of spatial correlation and imperfect channel estimation are taken into account.   

The rest of this paper is organised as follows. In Section~\ref{sec:sys}, we present the system model. The antenna selection algorithm is proposed in Section~\ref{sec:iid}. The proposed scheme is then extended to the spatially correlated channel model with imperfect channel estimation in Section~\ref{sec:sci}. Performance analysis of the proposed scheme and the relevant discussions are given in Section~\ref{sec:re} and \ref{sec:con} respectively.

\section{System Model}\label{sec:sys}
We consider an uplink system with a single user terminal (UT) with a single antenna and one BS with a large number of antennas. The number of antennas at the BS side is $M$, and the received signal vector is represented by
\begin{equation} \label{eq:y1}
\mathbf{y} = \mathbf{h}x + \mathbf{v},
\end{equation}
where $\mathbf{h} \in \mathbb{C}^{M \times 1}$ is the channel vector, and $x$ is the transmitted symbol with transmitted power ${\mathbb E} \left[xx^{*} \right] = \sigma_{x}^{2}$, where $\left(\cdot\right)^{*}$ denotes complex conjugate. In addition, at the receiver side, we introduce the additive white Gaussian noise (AWGN) vector $\mathbf{v} \in \mathbb{C}^{M \times 1}$, consisting of independent circularly symmetric complex Gaussian random variables with ${\mathbb E} \left[ \mathbf{v}\mathbf{v}^{H} \right] = \sigma_{v}^{2}\mathbf{I}_{M}$, where $\left(\cdot\right)^{H}$ denotes Hermitian transposition and $\mathbf{I}_{M}$ is the $M \times M$ identity matrix. Hence the signal-to-noise ratio (SNR) can be expressed as $\text{\textnormal{SNR}} = \sigma_{x}^{2}/\sigma_{v}^{2}$. Throughout this paper, we take into account the spatial correlation among the antennas and imperfect channel estimation,  described in the following. 
\subsection{Spatially Correlated Channel Model}
The spatially correlated channel $\mathbf{h}$ in the \eqref{eq:y1} can be characterised as following Kronecker model \cite{r16}
\begin{equation} \label{eq:H1}
\mathbf{h} = \mathbf{\Phi}_{r}^{1/2}\mathbf{h}_{i}\mathbf{\Phi}_{t}^{1/2},
\end{equation}
where the $\mathbf{h}_{i} \in \mathbb{C}^{M \times 1}$ is an uncorrelated complex channel vector whose entries are independent identically distributed (i.i.d) circularly symmetric complex Gaussian random variables with zero mean and unit variance. $\mathbf{\Phi}_{r}$ and $\mathbf{\Phi}_{t}$ determine the correlation between receiver antennas, and between transmitter antennas, respectively. Note that $\left(\cdot\right)^{1/2}$ in the \eqref{eq:H1} represents the Hermitian square root of a matrix. In the case of single antenna UT uplink transmissions, the correlation among the receiver antennas can be focused on. Notice that such assumption is valid for multiuser MIMO systems as well, since user terminals are autonomous \cite{MSP2013}. To this end, the spatially correlated channel vector can be given as
\begin{equation} \label{eq:H2}
\mathbf{h} = \mathbf{\Phi}_{r}^{1/2}\mathbf{h}_{i}.
\end{equation}

It is suggested that the exponential correlation model is a widely adopted approximation for the structure of the correlation matrix \cite{r16}, which can suitably evaluate the level of spatial correlation among antennas, 
as given by,
\begin{equation} \label{eq:r}
{\Phi}_{ij} = \left\{\begin{matrix}
\phi^{|j-i|}, & i \leq j \\ 
\left( \phi^{|j-i|}\right)^{*}, & i > j,
\end{matrix}\right.
\end{equation}
where ${\Phi}_{ij}$ is the entry of the receiver side correlation matrix $\mathbf{\Phi}_{r}$ and corresponds to the correlation between $i^{th}$ and $j^{th}$ receiver antenna. A single coefficient $\phi$ is also introduced, with $|\phi| \leq 1$, where, here and in \eqref{eq:r}, $|\cdot|$ denotes the absolute value operation. Hereafter we assume that the $M \times M$ correlation matrices $\mathbf{\Phi}_{r}$ is known, due to the fact that it is supposed to be less frequently varying than the channel matrix. Furthermore, the distribution of $\mathbf{h}_{i}$ is known to the receiver \cite{r17}, and $\mathbf{h}_{i}$ stays constant and is independent of the transmitted symbol $x$ and noise vector $\mathbf{v}$ during one transmission period.

\subsection{Imperfect Channel Estimation}
In practice, the channel is estimated at the receiver, by applying different channel estimation schemes such as MMSE-based pilot signalling estimation, which can introduce estimation errors. Since the correlation matrices are assumed to be available, the channel estimation can be applied for the uncorrelated channel component $\mathbf{h}_{i}$. The imperfect estimate $\hat{\mathbf{h}}_{i}$ of the $\mathbf{h}_{i}$ can be modelled as \cite{r11}
\begin{equation}
\hat{\mathbf{h}}_{i} = \sqrt{1-\tau}\mathbf{h}_{i} + \sqrt{\tau}\mathbf{e}_{i},
\end{equation} 
where $\mathbf{e}_{i}$ is the estimation error. It is suggested that $\mathbf{e}_{i}$ can be  independent of $\mathbf{h}_{i}$, due to the property of the MMSE estimator \cite{LargeMISO}, whose entries are i.i.d zero mean circularly symmetric complex Gaussian random variables. Here the estimation variance parameter $\tau \in [0,1]$ represents the estimation accuracy, i.e., $\tau = 1$ represents the extreme case that there is not correlation between the estimation of $\mathbf{h}_{i}$ and its actual value, whereas $\tau = 0$ corresponds to the perfect channel estimation without error\cite{LargeMISO}. Recalling \eqref{eq:H2}, the channel estimate $\hat{\mathbf{h}}$ can be further expressed as \cite{LargeMISO, r20}
\begin{align}
\label{eq:He1}
\hat{\mathbf{h}} &= \mathbf{\Phi}_{r}^{1/2}\hat{\mathbf{h}}_{i},\\
\label{eq:He2}
& = \sqrt{1-\tau}\mathbf{h} + \mathbf{e},
\end{align}
where $\mathbf{e} =  \sqrt{\tau}\mathbf{\Phi}_{r}^{1/2}\mathbf{e}_{i}$. Then, the effect of both antenna spatial correlation and imperfect channel estimation can be investigated, by adjusting the correlation coefficient $\phi$ and estimation variance parameter $\tau$.

\section{Multiple Antenna Selection Problem Formulation for the Uncorrelated Channel}
\label{sec:iid}
We first consider a single UT equipped with a single antenna at the transmitter side for the uncorrelated i.i.d channel network. 
In order to realise the multiple receiver antenna selection, here we introduce an antenna selection vector $\mathbf{h}_{s,i} \in \mathbb{C}^{M \times 1}$, which can also be considered as an equalisation vector, due to the fact that each receiver antenna is weighted by a corresponding channel coefficient in the vector $\mathbf{h}_{s,i}$. Considering the expression of received signal in \eqref{eq:y1}, the equalised signal can be given as $\hat{y}$, after we apply the antenna selection vector, as 
\begin{equation}
\hat{y} = \mathbf{h}_{s,i}^{H}(\mathbf{h}_{i}x + \mathbf{v}).
\end{equation}
Based on the equalised signal structure, the antenna selection can be obtained by minimising the mean squared error (MSE) at the receiver. To achieve this, we define the error signal as
\begin{align}
e & = x - \hat{y} \notag \\
& = x - \mathbf{h}_{s,i}^{H}(\mathbf{h}_{i}x+ \mathbf{v}).
\end{align}
By exploiting the structure of the error signal, the MSE can be formulated as
\begin{align}
\text{\textnormal{MSE}} & \coloneqq E[\left\| e \right\|^{2}] \notag\\ 
&= \sigma_{x}^{2} - \mathbf{h}_{s,i}^{H}\mathbf{h}_{i}\sigma_{x}^{2} - \sigma_{x}^{2}\mathbf{h}_{i}^{H}\mathbf{h}_{s,i} \notag\\
 &\qquad + \mathbf{h}_{s,i}^{H}\sigma_{x}^{2}\mathbf{h}_{i}\mathbf{h}_{i}^{H}\mathbf{h}_{s,i} + \mathbf{h}_{s,i}^{H}\sigma_{v}^{2}\mathbf{I}_{M}\mathbf{h}_{s,i},
\end{align}
where ``$\coloneqq$" is the definition sign. We then let 
\begin{align}
&\label{eq:hi} \tilde{\mathbf{h}}_{i} = \sigma_{x}^{2}\mathbf{h}_{i}, \\
&\label{eq:Ri} \mathbf{R}_{i} = \sigma_{x}^{2}\mathbf{h}_{i}\mathbf{h}_{i}^{H} + \sigma_{v}^{2}\mathbf{I}_{M}.
\end{align} 
Notice that $\mathbf{R}_{i}$ is positive definite, we apply Cholesky decomposition as $\mathbf{R}_{i} = \mathbf{L}_{i}\mathbf{L}_{i}^{H}$ where $\mathbf{L}_{i}$ is one $M \times M$ lower-triangular matrix. The expression of $\text{\textnormal{MSE}}$ can then be written as
\begin{align}
\text{\textnormal{MSE}} &=  \sigma_{x}^{2} - \mathbf{h}_{s,i}^{H}\mathbf{L}_{i}\mathbf{L}_{i}^{-1}\tilde{\mathbf{h}}_{i} \notag \\ &\qquad - \tilde{\mathbf{h}}_{i}^{H}\mathbf{L}_{i}^{-H}\mathbf{L}_{i}^{H}\mathbf{h}_{s,i}  + \mathbf{h}_{s,i}^{H}\mathbf{L}_{i}\mathbf{L}_{i}^{H} \mathbf{h}_{s,i}\notag\\ 
\label{eq:MSE2}
&= \sigma_{x}^{2} - \tilde{\mathbf{h}}_{i} ^{H}\mathbf{L}_{i}^{-H}\mathbf{L}_{i}^{-1}\tilde{\mathbf{h}}_{i}  + \left\| \mathbf{L}_{i}^{H}\mathbf{h}_{s,i} - \mathbf{L}_{i}^{-1}\tilde{\mathbf{h}}_{i} \right\|^{2}_{2}.
\end{align}

The only term in \eqref{eq:MSE2} related to the antenna selection vector $\mathbf{h}_{s,i}$, which can be further processed, is the last term, i.e., the $L_{2}$ norm (denoted by $\|\cdot\|_{2}$). Such a minimisation problem can be efficiently solved by using sparse approximation algorithms such as the orthogonal matching pursuit (OMP) algorithm, which has been shown that it outperforms the conventional SNR-based selection combining scheme in \cite{r3}. More specifically, since the vector $\mathbf{h}_{s,i}$ reflects the receiver antenna selection process, the only non-zero entries of $\mathbf{h}_{s,i}$ correspond to the selected receiver antenna (i.e., $\mathbf{h}_{s,i}$ becomes a sparsely structured vector). Hence, the acquisition of $\mathbf{h}_{s,i}$ transforms to a sparse approximation problem. We formulate this sparse approximation problem by generating a link between the sparse approximation and the MSE optimisation: the objective can be the minimisation of the $L_{2}$ norm, and the measurement dictionary and the target vector are $\mathbf{L}_{i}^{H}$ and $\mathbf{L}_{i}^{-1}\tilde{\mathbf{h}}_{i}$, respectively. In the OMP algorithm, an iterative calculation process is carried out to locate one column vector in the measurement dictionary that is the most correlated vector to the residual vector (which is generally initialised to be the target vector), at each iteration. One locally optimum solution is measured by solving a least-squared problem to update the residual vector. Here, for the sake of simplicity, we highlight the parameters in the algorithm relating to this work. The inputs of the OMP process are the measurement dictionary $\mathbf{L}_{i}^{H}$ and the target vector $\mathbf{L}_{i}^{-1}\tilde{\mathbf{h}}_{i}$, as well as a stopping criterion. Here the stopping criterion is selected as the desired number of iterations for the OMP algorithm, named $K_{s}$. We denote the proposed OMP algorithm as
\begin{align}
\label{eq:hsi}
\mathbf{h}_{s,i} & = \argmin_{\mathbf{h}_{s,i} |_{\text{\textnormal{OMP}}}}  \left\| \mathbf{L}_{i}^{H}\mathbf{h}_{s,i} - \mathbf{L}_{i}^{-1}\tilde{\mathbf{h}}_{i} \right\|_{2}, \notag \\
& \text{\textnormal{s.t.}} \left\| \mathbf{h}_{s,i} \right\|_{0} = K_{s},
\end{align}
where $\text{\textnormal{s.t.}}$ stands for ``subject to", $\| \cdot \|_{0}$ represents the $L_{0}$ norm, also informally the number of non-zero entries in a vector, and $\mathbf{h}_{s,i} |_{\text{\textnormal{OMP}}}$ refers to the value of $\mathbf{h}_{s,i}$ calculated by OMP algorithm. Notice that at the end of each iteration, the optimum solution is obtained, corresponding to one selection process of $\mathbf{h}_{s,i}$. Therefore, the stopping criterion $K_{s}$ also indicates the desired number of selected receiver antennas, and multiple antenna selection can be realised by using the sparsely structured antenna selection vector generated by the \eqref{eq:hsi}.

\section{Spatial Correlated Channel with Imperfect Channel Estimation}
\label{sec:sci}
In this section, we extend the OMP operation based antenna selection scheme taking into account spatial correlation among the antennas and imperfect channel estimation. 
Then, in order to use the OMP algorithm in \eqref{eq:hsi} to realise the multiple antenna selection in the spatially correlated channel, we generalise the expression of $\tilde{\mathbf{h}}_{i}$ in \eqref{eq:hi} and $\mathbf{R}_{i}$ in \eqref{eq:Ri} to 
\begin{align}
&\label{eq:hc} \tilde{\mathbf{h}} = \sigma_{x}^{2}\mathbf{h} = \sigma_{x}^{2}\mathbf{\Phi}_{r}^{1/2}\mathbf{h}_{i},\\
&\label{eq:Rc} \mathbf{R} = \sigma_{x}^{2}\mathbf{h}\mathbf{h}^{H} + \sigma_{v}^{2}\mathbf{I}_{M} = \sigma_{x}^{2}\mathbf{\Phi}_{r}^{1/2}\mathbf{h}_{i}\mathbf{h}_{i}^{H}\mathbf{\Phi}_{r}^{H/2} + \sigma_{v}^{2}\mathbf{I}_{M}.
\end{align}

The exponential correlation matrix $\mathbf{\Phi}_{r}$ in \eqref{eq:Rc} is accordingly a positive semidefinite matrix \cite{r17}, so it is necessary to verify the positive definiteness of $\mathbf{R}$ for its availability of Cholesky decomposition. To do so, we introduce the following lemma, and further define the $\mathbf{\Phi}_{r}$ as a symmetric real positive semidefinite matrix (i.e., $\phi \in [0,1)$).
\begin{lemma}
\label{pd}
Let $\mathbf{\Phi}_{r}$ be a symmetric real positive semidefinite matrix, and $\mathbf{R}_{h} = \mathbf{h}_{i}\mathbf{h}_{i}^{H}$ be a positive semidefinite matrix. Then $\mathbf{R}$ is positive definite.
\end{lemma}
\begin{proof}
Since $\mathbf{\Phi}_{r}$ is positive semidefinite, then its square root $\mathbf{\Phi}_{r}^{1/2}$ is positive semidefinite as well. In addition, $\mathbf{\Phi}_{r}$ is a symmetric real matrix, then $\mathbf{\Phi}_{r}$ is equal to its own conjugate transpose $\mathbf{\Phi}_{r}^{\dagger}$ (i.e., $\mathbf{\Phi}_{r}^{H}$), where $(\cdot)^{\dagger}$ represents conjugate transpose operation. Due to the property of the positive semi/definite matrix, it is easy to prove $\mathbf{\Phi}_{r}^{1/2}\mathbf{R}_{h}\mathbf{\Phi}_{r}^{H/2}$ is positive semidefinite, equivalently to 
\begin{equation}
\mathbf{z}^{\dagger}\mathbf{\Phi}_{r}^{1/2}\mathbf{R}_{h}\mathbf{\Phi}_{r}^{H/2}\mathbf{z} \geq \mathbf{0}, \forall \mathbf{z} \in \{ \mathbf{z} \in \mathbb{C}^{M \times 1} | \mathbf{z} \neq \mathbf{0}\}.
\end{equation}
Thus,
\begin{equation}
\mathbf{z}^{\dagger}(\mathbf{\Phi}_{r}^{1/2}\mathbf{R}_{h}\mathbf{\Phi}_{r}^{H/2} + \sigma_{v}^{2}\mathbf{I}_{M})\mathbf{z} \geq \mathbf{0}, \forall \mathbf{z} \in \{ \mathbf{z} \in \mathbb{C}^{M \times 1} | \mathbf{z} \neq \mathbf{0}\}.
\end{equation}
In this case, $\sigma_{v}^{2} > 0$, thus there exists no such one non-zero complex vector $\mathbf{z}_{e}$,  that $\mathbf{z}_{e} \in \{ \mathbf{z} \in \mathbb{C}^{M \times 1} | \mathbf{z} \neq \mathbf{0}\}$, let $\mathbf{z}_{e}^{\dagger}(\mathbf{\Phi}_{r}^{1/2}\mathbf{R}_{h}\mathbf{\Phi}_{r}^{H/2} + \sigma_{v}^{2}\mathbf{I}_{M})\mathbf{z}_{e} = \mathbf{0}$. Therefore,
\begin{equation}
\mathbf{z}^{\dagger}(\mathbf{\Phi}_{r}^{1/2}\mathbf{R}_{h}\mathbf{\Phi}_{r}^{H/2} + \sigma_{v}^{2}\mathbf{I}_{M})\mathbf{z}  > \mathbf{0}, \forall \mathbf{z} \in \{ \mathbf{z} \in \mathbb{C}^{M \times 1} | \mathbf{z} \neq \mathbf{0}\},
\end{equation}
which indicates that $\mathbf{R}$ is positive definite, as required.
\end{proof}

Based on Lemma~\ref{pd}, it can be proved that $\mathbf{R}$ in \eqref{eq:Rc} is positive definite, and the multiple antenna selection with the receiver side spatially correlated channel can be realised, by measuring revised sparse antenna selection vector $\mathbf{h}_{s,c}$, instead of $\mathbf{h}_{s,i}$ in  \eqref{eq:hsi}, and the relative components in the OMP algorithm. More specifically, we have the generalised $\tilde{\mathbf{h}}$ in \eqref{eq:hc} and $\mathbf{R}$ in \eqref{eq:Rc}, and $\mathbf{L}$ is the $M \times M$ lower-triangular matrix from Cholesky decomposed $\mathbf{R}$. Correspondingly, the measurement dictionary and the target vector become  $\mathbf{L}^{H}$ and $\mathbf{L}^{-1}\tilde{\mathbf{h}}$ respectively. We rewrite the structure of $\mathbf{h}_{s,c}$ as 
\begin{align}
\label{eq:hsc}
\mathbf{h}_{s,c}  & = \argmin_{\mathbf{h}_{s,c} |_{\text{\textnormal{OMP}}}}  \left\| \mathbf{L}^{H}\mathbf{h}_{s,c} - \mathbf{L}^{-1}\tilde{\mathbf{h}} \right\|_{2}, \notag \\
& \text{\textnormal{s.t.}} \left\| \mathbf{h}_{s,c} \right\|_{0} = K_{s},
\end{align}
by considering the same stopping criterion in the OMP operation as \eqref{eq:hsi}, i.e., the number of selected antennas. \*

Recall the Equation~\eqref{eq:He1} and \eqref{eq:He2}, we now consider the case with imperfect channel estimation. Under the same assumption of a single antenna UT uplink transmission, only the channel estimate vector $\hat{\mathbf{h}}$ is available to the receiver. Generalise the  $\tilde{\mathbf{h}}$ and $\mathbf{R}$ to $\tilde{\mathbf{h}}_{e}$ and $\mathbf{R}_{e}$, respectively, which can be given as  
\begin{equation}
\tilde{\mathbf{h}}_{e} = \sigma_{x}^{2}\hat{\mathbf{h}}
\end{equation} 
\begin{equation}
\mathbf{R}_{e} = \sigma_{x}^{2}\hat{\mathbf{h}}\hat{\mathbf{h}}^{H} + \sigma_{v}^{2}\mathbf{I}_{M}.
\end{equation}
In a similar way to that provided in Lemma~\ref{pd}, it is evident that the positive definiteness of $\mathbf{R}_{e}$ and the its availability of Cholesky decomposition can be satisfied. We can allocate the parameters for the OMP algorithm with imperfect channel estimation as
\begin{align}
\label{eq:hse}
\mathbf{h}_{s,e}  & = \argmin_{\mathbf{h}_{s,e} |_{\text{\textnormal{OMP}}}}  \left\| \mathbf{L}_{e}^{H}\mathbf{h}_{s,e} - \mathbf{L}_{e}^{-1}\tilde{\mathbf{h}}_{e} \right\|_{2}, \notag \\
& \text{\textnormal{s.t.}} \left\| \mathbf{h}_{s,e} \right\|_{0} = K_{s},
\end{align} 
where the $\mathbf{h}_{s,e}$ is the updated version of $\mathbf{h}_{s,c}$ in \eqref{eq:hsc} with consideration of channel estimation error, and $\mathbf{L}_{e}$ is the $M \times M$ lower-triangular matrix generated by the Cholesky decomposition of $\mathbf{R}_{e}$. Again, the stopping criterion is the desired number of selected antennas. 

\section{Simulation Results}
\label{sec:re}
In this section, we compare a series of bit error rate (BER) performances of our proposed scheme with MRC scheme. The system consisting of one single-antenna UT and one BS with a large number of antennas is considered. More specifically, we assume $M =$ 16, 64 or 128. BPSK modulation is applied in our simulations. The effect of sparsity of the antenna selection vector, antenna spatial correlation and imperfect channel estimation can be taken into account by adjusting the value of the parameter $K_{s}$, $\phi$ and $\tau$ in our programme.

\begin{figure}[!t]
\centering
\includegraphics[width=3.5in]{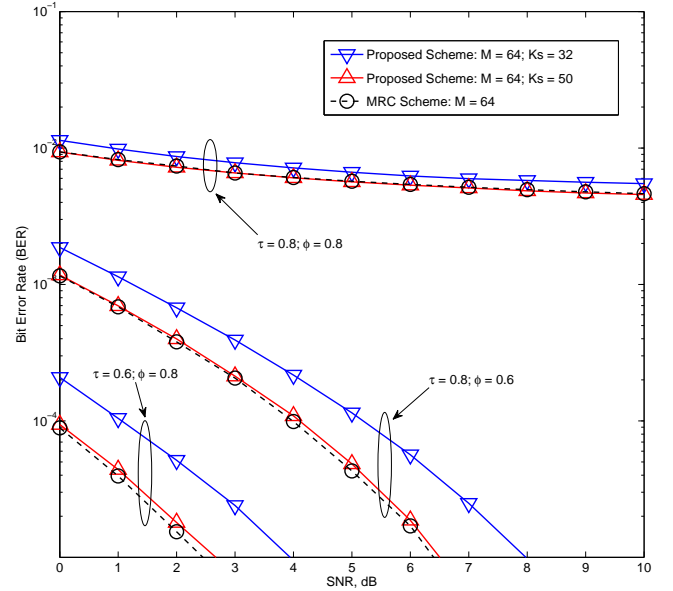}
\caption{BER versus SNR comparison between our proposed scheme and MRC scheme for a large number of receive antennas ($M =$ 64), with different levels of $K_{s}$, $\tau$ and $\phi$, and BPSK modulation.}
\label{fig:fig1}
\end{figure}

Fig.~\ref{fig:fig1} demonstrates the BER performance of the both schemes with different SNR per bit levels. The total number of BS antennas $M$ is set to 64, and correspondingly, we select the half number, i.e., $K_{s}$ equals to 32 out of 64, and more than half number of the BS antennas, i.e., $K_{s}$ is equal to 50 out of 64. Also, we examine several combinations of $\phi$ and $\tau$. It is not surprising to observe that the both schemes are considerably impacted by the high level of $K_{s}$, $\phi$ and $\tau$. However, due to the effective antenna selection process in our algorithms that can minimise the effect of highly correlated channels as well as the channel estimation error during the transmission, our proposed scheme with larger number of selected antennas (i.e., $K_{s}$ = 50) has nearly same performance as MRC, and the gap between the results of MRC and our method with only half antennas selected is fairly negligible. Notice that we show the case with high levels of antenna correlation and channel estimation error (e.g., $\phi$ and $\tau$ equal to 0.6 or even 0.8). In fact, such highly correlated channels can be experienced in our system since the very large BS antenna equipped. In addition, the high level of channel estimation error can be certainly introduced, due to the realistic transmission conditions such as limited feedback and high mobility of UT. 

\begin{figure}[!t]
\centering
\includegraphics[width=3.5in]{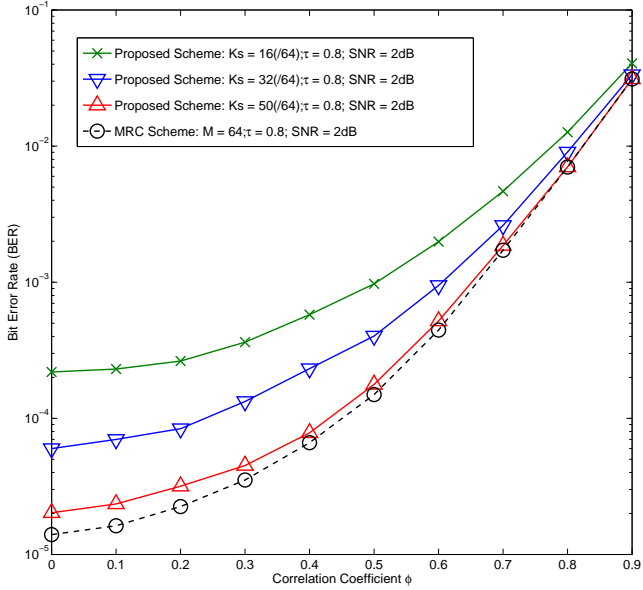}
\caption{BER versus $\phi$ performance comparison for our scheme and MRC with ($M =$ 64) and high estimation error (i.e., $\tau$ = 0.8), and different levels of $K_{s}$, in the low SNR regime (SNR = 2dB). BPSK applied.}
\label{fig:fig2}
\end{figure}

\begin{figure}[!t]
\centering
\includegraphics[width=3.5in]{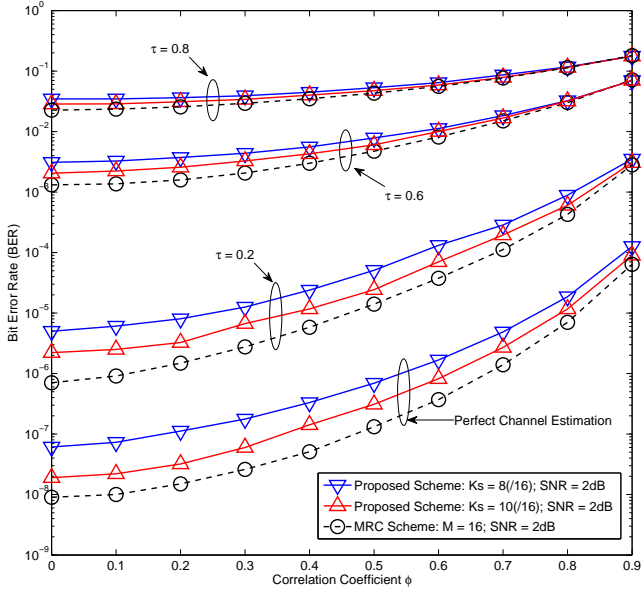}
\caption{BER versus $\phi$ performance comparison for our scheme and MRC with ($M =$ 16), and different levels of $K_{s}$ and $\tau$, in the low SNR regime (SNR = 2dB). BPSK applied.}
\label{fig:fig3}
\end{figure}

After the general observation of the performance in Fig.~\ref{fig:fig1}, now we focus on the effect of different combinations of $\tau$ and $\phi$, and the required number of selected antennas, shown in Fig.~\ref{fig:fig2} and Fig.~\ref{fig:fig3} respectively. First, Fig.~\ref{fig:fig2} illustrates the BER performance of the case, with $M$ = 64, $K_{s}$ = 16, 32 or 50, and $\tau$ = 0.8, by viewing a different aspect from Fig.~\ref{fig:fig1}, i.e., with different levels of $\phi$ and in the low SNR regime (SNR = 2dB). It is shown that our scheme has very similar performance with MRC, especially in the high region of $\phi$. In order to take a closer look of the performance with lower $\tau$, in Fig.~\ref{fig:fig3}, we choose a lower number of $M$, equals to 16, and select 8 or 10 antennas out of 16. The conclusion holds as well that the compared to the MRC, the performance of our proposed scheme is not degraded by combining only selected antennas, with high levels of $\tau$ and $\phi$ involving. Then, in the interest of high levels of antenna spatial correlation ($\phi$ = 0.8) and imperfect channel estimation ($\tau$ = 0.8), Fig.~\ref{fig:fig4} shows the BER performance versus the number of selected antenna $K_{s}$ of our scheme and MRC, with different levels of SNR. For the high SNR regime, the BER performance of our scheme is closely approached to that of MRC for $M$ = 64 is around 35. For the low SNR regime, approximately measuring, the required number of selected antenna $K_{s}$ is equal to 60 for $M$ = 128, or only 30 for $M$ = 64. It is suggested that when the bad transmission condition introduced in our system, e.g., low SNR regime and high levels of $\phi$ and $\tau$, our proposed scheme has similar, even identical performance as the MRC scheme, with less than half antennas selected, due to the effective selection process designed for different transmission situations.

\begin{figure}[!t] 
    \centering
    \subfigure[BER vs $K_{s}$; $M$ = 64]
    {
        \includegraphics[width=3.5in]{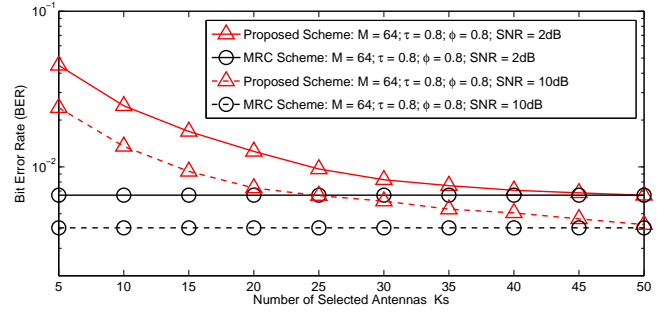}
        \label{fig:first_sub}
    }
    \\
    \subfigure[BER vs $K_{s}$; $M$ = 128]
    {
        \includegraphics[width=3.5in]{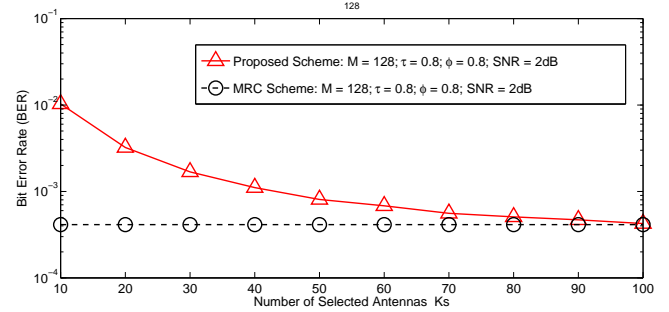}
        \label{fig:second_sub}
    }
\caption{BER versus $K_{s}$ comparison between our proposed scheme and MRC scheme for different number of receive antennas ($M=$ 64 or 128), with high $\tau$ and $\phi$ introduced for different SNR levels. BPSK applied.}
\label{fig:fig4}
\end{figure}


\subsection{Complexity Analysis}
The MRC algorithm requires a number of signal processing for entire diversity channels, which significantly increases the hardware complexity and cost due to the implementation of RF chains for all antennas in the massive MIMO system \cite{AS2013}. Instead, our proposed selection scheme allows the receiver to restore the signal to its original shape, only by weighting few (e.g., even less than the half number of antennas that shown in Fig.~\ref{fig:fig1} and \ref{fig:fig4}) selected channels with the sparsely structured antenna selection vector, and without degrading the system performance, which is a dramatic improvement in reducing the implementation overhead, e.g., the required number of RF chains, in practice. Consider the OMP algorithm presented in \eqref{eq:hsi}, \eqref{eq:hsc} and \eqref{eq:hse}, the input components are based on the channel estimation, which can be physically performed on each antenna with a less complex device rather than the full transceiver \cite{AS2013}. Then the antenna selection can be realised by using the output vector, i.e., the $M\times1$-dimensional antenna selection vector with only $K_{s}$ nonzero elements. In addition, the iteration times is equal to the stopping criterion $K_{s}$. Hence, the computational complexity of the OMP algorithms is $\mathcal{O}(K_{s}^{2}M)$.

\section{Discussion}
\label{sec:con}
Throughout this work, we proposed a new antenna selection scheme for the single-user massive MIMO uplink transmission by applying the sparsely structured antenna selection vector, and then generalised our proposed scheme with the consideration of spatial correlation and imperfect channel estimation. Numerical simulation results show that when the severe transmission condition is experienced in our system, such as very low SNR regime, highly correlated channel and considerable estimation error, our proposed scheme has closely approached performance as the well-adopted MRC scheme, but requiring few selected antennas, due to the effective selection process by applying the sparsely structured antenna selection vector, which can significantly reduce the implementation overhead.\*

Furthermore, due to space limitations, we present our system model here as single user systems, and it is being considered to emphasise upon MU-MIMO in the journal version of this work, e.g., the extension to a single-cell multiuser scenario can be achieved by considering that users independently transmit data to the base station.


\section*{Acknowledgment}
This work has been done within joint project, supported by Huawei Tech. Co., Ltd, China.




\bibliographystyle{IEEEtran}

\bibliography{IEEEabrv,DMref}

%




\end{document}